\documentclass[letterpaper, 10pt, conference]{ieeeconf}
\IEEEoverridecommandlockouts 
\usepackage{amssymb}
\usepackage{amsmath}

\usepackage{amsthm}
\usepackage{cite}
\usepackage{graphics}
\usepackage{hyperref}
\usepackage{multirow}
\usepackage{tikz}

\usetikzlibrary{shapes,arrows}
\tikzstyle{block} = [draw, rectangle]
\tikzstyle{sum}   = [draw, circle]
\tikzstyle{point} = [coordinate]

\theoremstyle{definition}
\newtheorem{defn}{Definition}
\newtheorem{theorem}{Theorem}
\newtheorem{lemma}{Lemma}[theorem]
\newtheorem{prop}{Proposition}
\newtheorem{corollary}{Corollary}[theorem]

\def\R{\mathbf{\Phi}_x}
\def\M{\mathbf{\Phi}_u}
\def\Rc{\mathbf{R_c}}
\def\Mc{\mathbf{M_c}}
\def\Rhat{\hat{\mathbf{\Phi}}_x}
\def\Mhat{\hat{\mathbf{\Phi}}_u}
\def\Rtilde{\tilde{\mathbf{\Phi}}_x}
\def\Mtilde{\tilde{\mathbf{\Phi}}_u}

\def\PhiStack{\begin{bmatrix}\R \\ \M\end{bmatrix}}
\def\PhicStack{\begin{bmatrix}\Rc \\ \Mc\end{bmatrix}}
\def\ZAB{\begin{bmatrix}zI-A && -B\end{bmatrix}}

\def\Kone{\mathbf{K_1}}

\def\D{\mathbf{\Delta}}
\def\Dc{\mathbf{\Delta_c}}

\def\HTwo{\mathcal{H}_{2}}
\def\LOne{\mathcal{L}_{1}}
\def\K{\mathbf{K}}

\newif\ifshowWriterComment
\newcommand\writercomment[3]{\expandafter\newcommand\csname #2\endcsname[1]{\ifshowWriterComment{\color{#3} (#1: ##1)}\fi}}

\writercomment{Lisa}{JSL}{green}

\showWriterCommenttrue

\title{\LARGE \bf
	Separating Controller Design from Closed-Loop Design:\\
	A New Perspective on System-Level Controller Synthesis
}

\author{Jing Shuang (Lisa) Li and Dimitar Ho
	\thanks{Authors are with the Department of Computing and Mathematical Sciences, California Institute of Technology.
		{\tt\small jsli@caltech.edu},
		{\tt\small dho@caltech.edu}
	}%
}
{\tiny }
\begin{document}	
	\maketitle
	\thispagestyle{empty}
	\pagestyle{empty}
	
	\begin{abstract}
	We show that given a desired closed-loop response for a system, there exists an affine subspace of controllers that achieve this response. By leveraging the existence of this subspace, we are able to separate controller design from closed-loop design by first synthesizing the desired closed-loop response and then synthesizing a controller that achieves the desired response. This is a useful extension to the recently introduced \textit{System Level Synthesis} framework, in which the controller and closed-loop response are jointly synthesized and we cannot enforce controller-specific constraints without subjecting the closed-loop map to the same constraints. 
	
	We demonstrate the importance of separating controller design from closed-loop design with an example in which communication delay and locality constraints cause standard SLS to be infeasible. Using our new two-step procedure, we are able to synthesize a controller that obeys the constraints while only incurring a 3\% increase in LQR cost compared to the optimal LQR controller.
\end{abstract}
	
	\section{INTRODUCTION} \label{sec:introduction}

Large-scale distributed cyberphysical systems (e.g. power grids, intelligent transportation systems) are composed of numerous local controllers that exchange local information via some communication network. The information that each local controller is able to obtain is limited by properties of the communication network, e.g. delay. It is a challenge to scalably synthesize optimal local controllers subject to the limitations of the communication network \cite{Ho1971,Mahajan2012, Rotkowitz2005, Bamieh2002, Bamieh2005, Nayyar2013}.

The recently developed \textit{System Level Synthesis} (SLS) framework addresses this challenge by shifting the optimization from the space of available controllers to the space of achievable system closed-loop maps \cite{Anderson2019}. In doing so, it allows the problem to be decomposed into sub-problems to be solved in parallel, resulting in a synthesis procedure with $O(1)$ complexity \cite{Wang2018}.

In the original SLS framework, the closed-loop maps themselves are used to implement the controller, and thus any constraints applied to the controller are directly enforced on the closed-loop response as well. However, the abovementioned communication limitations motivate constraints on \textit{controllers}, not closed-loop maps; by applying these constraints on the closed-loop response, we unnecessarily limit the space over which we can search for solutions.

Standard SLS is infeasible under excessive communication constraints. \cite{Matni2018} addresses this by searching over approximate closed-loop maps instead of exact closed-loop maps; constraints are imposed on the approximate closed-loop maps. We propose an alternative two-step procedure, as follows:
\begin{enumerate}
	\item Synthesize the desired closed-loop response, subject to closed-loop constraints. This can be done using SLS or any other linear synthesis method (Proposition \ref{prop:all_linear_ctrllers})	
	\item Synthesize the controller, subject to controller constraints
\end{enumerate}

To fully separate closed-loop map constraints from controller constraints, we require a controller that is implemented using transfer matrices \textit{other} than the closed-loop maps. We define the space of such matrices in Theorem \ref{thm:cl_impl_matrices} and give conditions for their existence in Lemma \ref{lemma:existence_solutions}.

The main contribution of this paper is to introduce the controller synthesis step of the design procedure and demonstrate its importance. We show that our proposed two-step synthesis allows us to design low-cost, distributed controllers that were unavailable to us in the previous framework. Additionally, the controller synthesis problem can be decomposed into parallelizable sub-problems, much like the original SLS problem.
	\section{PRELIMINARIES} \label{sec:preliminaries}
\subsection{Notation}
We use italicized lower-case letters (e.g. $x_t$) to denote vectors in the time domain. We use italicized upper-case letters (e.g. $A$) to denote constant matrices. We use superscripts to denote individual matrix elements (e.g. $A^{i,j}$). 

We use boldface lower and upper case letters (eg. $\mathbf{x}$, $\R$, $\Rc$) to denote signals and transfer matrices in the frequency domain. We use $R_c(k)$ to denote the $k$th spectral component of $\Rc$, i.e. $\Rc(z) = \sum_{k=0}^{\infty} R_c(k)z^{-k}$.

In this paper, we will restrict ourselves to strictly proper finite-impulse-response (FIR) transfer matrices, i.e. $\Rc(z) = \sum_{k=1}^{T} R_c(k)z^{-k}$, $T \in \mathbb{Z}_+$.

\subsection{System setup}
We use the same setup as in (2.1) of \cite{Anderson2019}:

\begin{equation} \label{eq:dynamics}
x_{t+1} = Ax_t + Bu_t + w_t
\end{equation} 

where $x$, $w$ $\in \mathbb{R}^n$ and $u$ $\in \mathbb{R}^m$. In this paper we focus on the time-invariant case (i.e. $A$, $B$ have no time-dependence) with state feedback.

$\R$ and $\M$ are the closed-loop maps from $w$ to $x$ and $u$, with FIR time horizon $T$:
\begin{equation} \label{eq:clmaps_freq}
\begin{bmatrix}\mathbf{x} \\ \mathbf{u} \end{bmatrix} = \begin{bmatrix}\R \\ \M \end{bmatrix}\mathbf{w}    
\end{equation}

\subsection{Controller implementation}

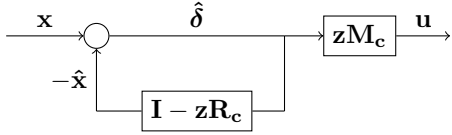
\begin{figure} 
\centering
\begin{tikzpicture}
    \node [point, name=input] {};
    \node [sum, right of=input, node distance=1.2cm] (sum) {};
    \node [point, right of=sum, node distance=2.5cm] (pt1) {};
    \node [block, right of=sum, node distance=3.5cm] (zM) {$\mathbf{zM_c}$};
    
    \node [point, right of=zM, node distance=1.2cm, name=output] {};

    \draw [->] (input) -- node[above] {$\mathbf{x}$} (sum);
    \draw [-] (sum) -- node[above, name=delta] {$\boldsymbol{\hat{\delta}}$} (pt1);
    \draw [->] (pt1) -- node[] {} (zM);
    \draw [->] (zM) -- node[above] {$\mathbf{u}$} (output);
    
    \node [point, below of=sum, node distance=1cm] (pt2) {};
    \node [point, below of=pt1, node distance=1cm] (pt3) {};
    \node [block, below of=delta, node distance=1.29cm] (IzR) {$\mathbf{I-zR_c}$};
    \draw [-] (pt1) -- node {} (pt3); 
    \draw [-] (pt3) -- node {} (IzR);
    \draw [-] (IzR) -- node {} (pt2);
    \draw [->] (pt2) -- node[left] {$\mathbf{-\hat{x}}$} (sum);
\end{tikzpicture}
\caption{Implementation of state feedback controller}
\label{fig:blockdiag}
\end{figure}
Fig. \ref{fig:blockdiag} shows the controller implementation. $\Rc$ and $\Mc$ are the implementation matrices, with order (i.e. FIR time horizon) $T_c$.

The controller includes two internal signals; $\mathbf{\hat{x}}$ and $\boldsymbol{\hat{\delta}}$. The equations describing the controller are
\begin{subequations} \label{eq:impl_eqns_time}
\begin{equation} \label{eq:impl_delta_time}
    \hat{\delta}_t = x_t - \sum_{k=2}^{T_c} R_c(k) \hat{\delta}_{t-k+1}
\end{equation}
\begin{equation} \label{eq:impl_u_time}
    u_t = \sum_{k=1}^{T_c} M_c(k) \hat{\delta}_{t-k+1}
\end{equation}
\end{subequations}
where (\ref{eq:impl_delta_time}) assumes that $R_c(1)$ is the identity. For a more detailed derivation, refer to \cite{Ho2019}. The corresponding frequency-domain equations are
\begin{subequations} \label{eq:impl_eqns}
\begin{equation} \label{eq:impl_delta_freq}
    \boldsymbol{\hat{\delta}} = \mathbf{x+(I-zR_c)}\boldsymbol{\hat{\delta}}
\end{equation} 
\begin{equation} \label{eq:impl_x_freq}
    \mathbf{x} = z\Rc\boldsymbol{\hat{\delta}}
\end{equation} 
\begin{equation} \label{eq:impl_u_freq}
    \mathbf{u} = z\Mc\boldsymbol{\hat{\delta}}
\end{equation} 
\end{subequations}

\begin{prop} \label{prop:all_linear_ctrllers}
Any linear controller (i.e. $\mathbf{u}=\mathbf{K}\mathbf{x}$) can be implemented using the controller structure defined in Fig. \ref{fig:blockdiag}.
\end{prop}
\begin{proof}
We can construct closed-loop maps $\R$ and $\M$ directly from $\mathbf{K}$, as shown in (4.4) of \cite{Anderson2019}:
\begin{subequations} \label{eq:phi_definitions}
\begin{equation}
    \R = (zI-A-B\mathbf{K})^{-1}
\end{equation}
\begin{equation}
    \M = \mathbf{K}(zI-A-B\mathbf{K})^{-1}
\end{equation}
\end{subequations}
We can then set $\Rc=\R$ and $\Mc=\M$ in (\ref{eq:impl_eqns}), which gives back the original controller $\mathbf{u}=\mathbf{K}\mathbf{x}$. \qedsymbol
\end{proof}

	\section{IMPLEMENTATION MATRICES} \label{sec:clim}

\subsection{Controllers and closed-loop maps}
\begin{theorem} \label{thm:unique_k}
	Let ($\R$, $\M$) be stable closed-loop maps. The only linear controller $\K$ (i.e. $\mathbf{u}=\K\mathbf{x}$) that achieves these closed-loop maps is $\K=\M\R^{-1}$.
\end{theorem}
\begin{proof}
	By Theorem 4.1 in \cite{Anderson2019}, $\K=\M\R^{-1}$ achieves the closed-loop maps. We show uniqueness by contradiction. Assume there is another linear controller $\Kone$, $\Kone \neq \K$, that also achieves the desired closed-loop maps. Since both $\K$ and $\Kone$ achieve ($\R$, $\M$),
	
	\begin{subequations}
		\begin{equation} \label{eq:r_k1k2}
			\R = (zI - A - B\Kone)^{-1} = (zI - A - B\K)^{-1}
		\end{equation}
		\begin{equation} \label{eq:m_k1k2}
			\M = \Kone(zI - A - B\Kone)^{-1} = \K(zI - A - B\K)^{-1}
		\end{equation}
	\end{subequations}
	Substituting (\ref{eq:r_k1k2}) into (\ref{eq:m_k1k2}) gives
	\begin{equation}
		\Kone\R = \K\R
	\end{equation}
	Since $\R$ is invertible, this implies that $\Kone=\K$. Contradiction! \qedsymbol
\end{proof}

Theorem \ref{thm:unique_k}, along with the definitions from (\ref{eq:phi_definitions}), show a one-to-one mapping between ($\R$, $\M$) and $\K$. However, the linear controller $\K$ can be implemented in a variety of ways. For example, we could directly implement $\mathbf{u}=\K\mathbf{x}$; we could also implement a linear controller using the structure shown in Figure \ref{fig:blockdiag}. In the original SLS framework, the latter is used to avoid direct matrix inversion of $\R$.

\subsection{Implementing closed-loop maps}
For the controller structure defined in Fig. \ref{fig:blockdiag}, let the controller implemented by ($\Rc$, $\Mc$) achieve closed-loop maps ($\Rtilde$, $\Mtilde$). We define the following terminology:

\begin{defn}
	($\Rc$, $\Mc$) are the \textit{implementation transfer matrices} for the closed-loop maps ($\Rtilde$, $\Mtilde$). We will refer to them as \textit{implementation matrices}.
\end{defn}

\begin{defn}
	We call ($\Rtilde$, $\Mtilde$) the \textit{implemented closed-loop maps} of the controller ($\Rc$, $\Mc$).
\end{defn}

The implemented closed-loop maps are found by combining (\ref{eq:impl_eqns_time}) and (\ref{eq:dynamics}) as done in \cite{Ho2019}:
\begin{equation} \label{eq:impl_cl_maps}
	\begin{bmatrix}\Rtilde \\ \Mtilde\end{bmatrix} = \PhicStack \Dc^{-1}
\end{equation}

Where $\Dc$ is a helper variable defined as
\begin{equation} \label{eq:Delta_c_freq}
\Dc = \ZAB \PhicStack
\end{equation}
Note that $\Dc$ can also be written as $I+\D$. This is the same formulation used by (4.22) in \cite{Anderson2019}, modulo notational differences (we use $\Rc$ and $\Mc$ instead of $\Rhat$, $\Mhat$). $\Dc$ is invertible since its leading spectral element, $I$, is invertible.

Our analysis largely focuses on closed-loop maps ($\R, \M$) instead of the controller $\K$. However, due to the one-to-one mapping between controller and closed-loop maps, we can also view ($\Rc$, $\Mc$) as implementation matrices for the controller $\K=\M\R^{-1}$.

\begin{theorem} \label{thm:cl_impl_matrices}
	For $R_c(1)=I$, ($\Rc$, $\Mc$) are implementation matrices for ($\R$, $\M$) \textit{if and only if} they satisfy
	\begin{equation} \label{eq:implicit_constr}
		\PhicStack = \PhiStack \ZAB \PhicStack
	\end{equation}
\end{theorem}
\begin{proof}
	\textit{Necessity}. If ($\Rc$, $\Mc$) are implementation matrices for ($\R$, $\M$), then we require
	\begin{equation} \label{eq:equal_maps}
		\begin{bmatrix} \Rtilde \\ \Mtilde \end{bmatrix} = \PhiStack
	\end{equation}
		
	Substituting (\ref{eq:impl_cl_maps}) into (\ref{eq:equal_maps}) and multiplying by $\Dc$, then writing out $\Dc$ in terms of ($A$, $B$, $\Rc$, $\Mc$), gives (\ref{eq:implicit_constr}).
	
	\textit{Sufficiency}. If ($\Rc$, $\Mc$) satisfy (\ref{eq:implicit_constr}), we can substitute (\ref{eq:implicit_constr}) into (\ref{eq:impl_cl_maps}) to conclude that ($\Rtilde$, $\Mtilde$) = ($\R$, $\M$), i.e. ($\Rc$, $\Mc$) are implementation matrices for ($\R$, $\M$). \qedsymbol
\end{proof}

This constraint describes an affine subspace of implementation matrices for ($\R$, $\M$).

\begin{corollary} \label{corollary:m1}
	If ($\Rc$, $\Mc$) are implementation matrices for ($\R$, $\M$), then the first spectral components of $\M$ and $\Mc$ are equal, i.e. $M_c(1)$ = $\Phi_u(1)$.
\end{corollary}
This equivalence arises directly from writing (\ref{eq:implicit_constr}) in terms of its spectral elements.

\begin{corollary} \label{corollary:self_clim}
	For $T_c \geq T$, ($\R$, $\M$) are implementation matrices for themselves.
\end{corollary}

($\R$, $\M$) are used as implementation matrices in \cite{Anderson2019}.

\begin{corollary}
	If ($\Rc$, $\Mc$) are implementation matrices for ($\R$, $\M$), then $K = \M\R^{-1} = \Mc\Rc^{-1}$
\end{corollary}

\subsection{Existence of solutions}

To better understand the dimension of the space of implementation matrices, we rearrange the constraint (\ref{eq:implicit_constr}) so that the variables ($\Rc$, $\Mc$) appear on only one side of the constraint.

Rewrite $\Dc$ in block-matrix form:
\begin{equation}
    \begin{bmatrix}
        \Delta_c(0) \\ \Delta_c(1) \\ \vdots \\ \Delta_c(T_c)
    \end{bmatrix} = 
    \begin{bmatrix}
        I & & & 0 & \\
        -A & I & & -B & & \\
         & \ddots & \ddots & & \ddots & \\
         & & -A & & & -B
    \end{bmatrix}
    \begin{bmatrix}
        R_c(1) \\ \vdots \\ R_c(T_c) \\ M_c(1) \\ \vdots \\ M_c(T_c)
    \end{bmatrix}
\end{equation}

Rewrite the right hand side of (\ref{eq:implicit_constr}) in block-matrix form:
\begin{equation}
    \begin{bmatrix}
        R_c(1) \\ \vdots \\ \vdots \\ R_c(T_c) \\ 0 \\ \vdots \\ 0
    \end{bmatrix} = 
    \begin{bmatrix}
        \Phi_x(1) \\
        \Phi_x(2) & \ddots \\
        \vdots \\
        \Phi_x(T) \\
         & \ddots \\
         & & \Phi_x(T)
    \end{bmatrix}
    \begin{bmatrix}
        \Delta_c(0) \\ \Delta_c(1) \\ \vdots \\ \Delta_c(T_c)
    \end{bmatrix} 
\end{equation}
We show only the formulation for $\Rc$; the formulation for $\Mc$ is identical but with $\M$ and $\Mc$ instead of $\R$ and $\Rc$.

Using the block-matrix formulations, we can rearrange (\ref{eq:implicit_constr}) into a constraint of the form 

\begin{subequations} \label{eq:explicit_constr}
	\begin{equation} \label{eq:constr_FvG}
	    Fv = G
	\end{equation}
	\begin{equation}
		v = \begin{bmatrix} R_c(2) \\ \vdots \\ R_c(T_c) \\ M_c(1) \\ \vdots \\ M_c(T_c) \end{bmatrix}
	\end{equation}
\end{subequations}

where $F$ and $G$ are matrices that do not depend on $\Rc$ and $\Mc$. The total number of constraints is $(T_c+T)(m+n)$. 

\begin{lemma} \label{lemma:existence_solutions}
	The implementation constraints (as defined in (\ref{eq:implicit_constr})) are feasible \textit{if and only if} $\mathrm{rank}(F) = \mathrm{rank}(F | G)$. If feasible, the solution space has dimension $dim(\mathrm{null}(F)) \times n$, where $n$ is the number of states in the system.
\end{lemma}
\begin{proof}
	This result is a direct application of the Rouch\'e-Capelli theorem to the linear system defined in (\ref{eq:explicit_constr}). \qedsymbol
\end{proof}

Corollary \ref{corollary:self_clim} states that (\ref{eq:implicit_constr}) has at least one solution for $T_c \geq T$. When $T_c < T$, we can check the rank of $F$ and $[F | G]$ and calculate the dimension of the solution space if it exists.
	\section{STABILITY} \label{sec:stability}

\subsection{Internal dynamics} \label{sec:int_dynamics}
The system is internally stable if the dynamics of $\hat{\delta}$, the internal signal, are stable. By substituting (\ref{eq:impl_eqns_time}) into (\ref{eq:dynamics}) and rearranging, we can obtain internal dynamics of the form

\begin{subequations}
\begin{equation} 
	z_t =
	\begin{bmatrix}
	\hat{\delta}_{t-T_c+1} \\
	\vdots \\
	\hat{\delta}_{t-1} \\
	\hat{\delta}_t
	\end{bmatrix}, \quad z_{t+1} = A_zz_t
\end{equation}
\begin{equation}
	A_z =
    \begin{bmatrix}
    0 & I & \ldots & 0 & 0 \\
    \vdots & & & \ddots \\
     0 & 0 & \ldots & 0 & I \\
    -\Delta_c(T_c) & & \ldots & & -\Delta_c(1)
    \end{bmatrix}
\end{equation} 
\end{subequations}

\subsection{Stability check} \label{sec:stability_check}
We can verify internal stability \textit{a posteriori} by checking that $A_z$
is stable. Alternatively, a sufficient condition for internal stability is $\|\D\| < 1$ \cite{Anderson2019}.

The stability of $A_z$ can be checked in a distributed manner. First, a helpful proposition:

\begin{prop} \label{prop:distr_check}
	Let $\|\cdot \|$ be an induced matrix norm. For $A \in \mathbb{R}^{n \times n}$, if $\exists m > 0$ s.t. $\|A^m\| < 1$, then $A$ is stable.
\end{prop}
\begin{proof}
	Let $\rho = \|A^m\|^{1/m}$, $\rho \in [0, 1)$. Using norm submultiplicativity and some algebra, we can show that $\forall t > m$, $\|A^t\| \leq C\rho^t$ where $C$ is some constant. Using this upper bound and induced norm properties, we can show that $\forall x_o \in \mathbb{R}^n$, $\lim_{t\to\infty}\|A^tx_o\| = 0$. This is the definition of stability in the discrete time setting. \qedsymbol
\end{proof}

Let each processor store $A_z$ and some columns of $A_z^k$, denoted $A_{z(i:j)}^k$. Overall, every column of $A_z^k$ is stored on some processor. The stability check procedure is as follows, starting with $k=1$:

\begin{enumerate}
	\item Calculate $A_{z(i:j)}^k$ by multiplying $A_z$ and $A_{z(i:j)}^{k-1}$
	\item Check the induced 1-to-1 norm of $A_{z(i:j)}^k$
	\item Consensus on whether a termination condition has been met. If no termination condition is met, increment $k$ and return to Step 1
\end{enumerate}

The clear termination condition is $\|A_z^k\| < 1$; then, $A_z$ is certified to be stable by Proposition \ref{prop:distr_check}. We suggest two additional termination conditions:

\begin{itemize}
	\item $\|A_z^k\| > M$, where $M$ is some predetermined threshold. Since $\|A_z^k\|$ corresponds to the amplitude of the transient response, this termination condition corresponds to finding an unacceptably large transient condition
	\item $k > k_{max}$, where $k_{max}$ is some predetermined maximum number of iterations
\end{itemize}

Both conditions would indicate that the stability check failed to certify stability. Since we select a column-wise separable norm, the entire procedure can be distributed. The complexity per iteration scales quadratically with $n$, under the conservative assumption that each node has at least one processor. For the system in Section \ref{sec:examples}, this procedure certifies stability in 7 iterations for the low-order controller and 32 iterations for the full-order controller.

	\section{APPROXIMATE IMPLEMENTATIONS} \label{sec:aprox_impl}

The solution space defined by (\ref{eq:implicit_constr}), although it exists for $T_c \geq T$, often yields solutions that are unstable. Further, Corollary \ref{corollary:m1} gives a fundamental limit on the sparsity of $\Mc$. If $\Phi_u(1)$ is dense, we cannot find implementation matrices that support any type of sparsity (e.g. communication delay, locality). These necessitate relaxations of (\ref{eq:implicit_constr}).

For a relaxed implementation, we want the implemented closed-loop maps ($\Rtilde$, $\Mtilde$) to be as close to the optimal closed-loop maps ($\R$, $\M$) as possible while maintaining internal stability, i.e.

\begin{equation} \label{eq:noncvx_relax}
	\begin{aligned}
	\min_{\Rc, \Mc} \|\PhicStack (I + \D)^{-1} - \PhiStack\|\\
	\textrm{s.t.} \quad (I + \D)^{-1} \textrm{stable}, \PhicStack \in \mathcal{S}
	\end{aligned}
\end{equation}

where $\mathcal{S}$ includes sparsity and FIR constraints, and $I + \D = \Dc$. This optimization problem is clearly nonconvex. Factoring the objective function as 
\begin{equation}
	\|(\PhicStack - \PhiStack(I + \D))(I + \D)^{-1}\|
\end{equation}
and using similar submultiplicativity, small-gain, and power series arguments as Section 4.5.1 of \cite{Anderson2019}, we can upper bound the optimization problem (\ref{eq:noncvx_relax}) with this quasi-convex problem:
\begin{equation} \label{eq:nested_relax}
	\begin{aligned}
	\min_{\gamma\in[0,1)}\frac{1}{1-\gamma} \min_{\Rc, \Mc, \D} \|\PhicStack - \PhiStack(I + \D)\|\\
	\textrm{s.t.} \quad \ZAB\PhicStack=(I + \D), \\
	\|\D\| \leq \gamma, \PhicStack \in \mathcal{S}
	\end{aligned}
\end{equation}

This is similar to the virtualized SLS method \cite{Matni2018} \cite{Anderson2019}, with one key difference. For an objective $g(\R, \M)$, the virtualized SLS method uses $g(\Rc, \Mc)$ as the objective, while our two-step method uses
\begin{equation}
	\|\PhicStack - \PhiStack(I + \D)\| 
\end{equation}
as the objective. This is the equation error for (\ref{eq:implicit_constr}), and is a heuristic for the closed-loop difference.

The nested optimization problem defined by (\ref{eq:nested_relax}) is time-consuming to solve; it can also be mathematically infeasible if the sparsity constraints $\mathcal{S}$ are too strict. We instead solve (\ref{eq:actual_relax}), which is much quicker and uses a regularizer on $\D$ to promote stability. We suggest starting with a small $\lambda$, solving (\ref{eq:actual_relax}), checking for stability using the distributed method presented in Section \ref{sec:stability_check}, and increasing $\lambda$ if the stability check is failed. Alternatively, we can enforce $\|\D\| < 1$.

\begin{equation} \label{eq:actual_relax}
	\begin{aligned}
	\min_{\Rc, \Mc, \D} \|\PhicStack - \PhiStack(I + \D)\| + \lambda\|\D\|\\
	\textrm{s.t.} \quad \ZAB\PhicStack=(I + \D), \PhicStack \in \mathcal{S}
	\end{aligned}
\end{equation}

We can also include additional objectives in (\ref{eq:actual_relax}), e.g. $\LOne$ regularization on ($\Rc$, $\Mc$) to promote sparsity.

The optimization problem (\ref{eq:actual_relax}) is column-wise separable if we choose a column-wise separable norm for the objective (e.g. $\HTwo$ norm). Like the original SLS problem, it can be decomposed into subproblems to be solved in parallel.
	\section{CLOSED-LOOP CONSTRAINTS VS. CONTROLLER CONSTRAINTS} \label{sec:cl_vs_ctrller}
In this section, we discuss the physical interpretation of separately applying locality and delay constraints to the closed-loop and to the controller, and when such constraints are appropriate. This separation is not possible in standard SLS, since the closed-loop maps themselves are used as implementation matrices for the controller.

First, a result on how applying controller constraints on the closed-loop maps can be overly restrictive:
\begin{lemma} \label{lemma:m_in_k}
	Let $\K$ be the controller corresponding to the closed-loop maps ($\R$, $\M$). Then, the operator $\M$ lies in the range of the operator $\K$.
\end{lemma}
\begin{proof}
	By Theorem \ref{thm:unique_k}, we have that $\K\R$ = $\M$.
\end{proof}

Lemma \ref{lemma:m_in_k} shows that sparsity constraints (e.g. locality, delay) on $\K$ will translate to sparsity constraints on $\M$, but not $\R$; directly applying these constraints on $\R$ may be too restrictive. Note that although it is also true that $\K\Rc$ = $\Mc$, both $\Mc$ and $\Rc$ must obey sparsity constraints as they are directly used in the implementation.

\subsection{Locality}
Let $\mathcal{L}(i)$ denote the locality of node $i$. Generally, $\mathcal{L}(i)$ consists of the $l$ closest neighbours of node $i$ in the network. Locality constraints restrict spectral components of $\Rc$ and $\Mc$ (or $\R$ and $\M$) to have nonzero support only over the allowed localities; i.e. 
\begin{equation} \label{eq:locality_constr}
\begin{aligned}
    R_c(k)^{i,j} = 0 \quad \forall j \notin \mathcal{L}(i) \\
    BM_c(k)^{i,j} = 0 \quad \forall j \notin \mathcal{L}(i) \\
\end{aligned}
\end{equation}
where $B$ is the actuation matrix of the system.

For a system with nodes arranged in a chain configuration and $\mathcal{L}(i)$ equal to the $l$ closest neighbours of node $i$, these constraints result in banded diagonal $R_c(k)$ and $M_c(k)$ with a band width of $2l+1$ $\forall k$.

When we apply locality constraints on the implementation matrices as per (\ref{eq:locality_constr}), we enforce that node $i$ will only communicate with nodes in $\mathcal{L}(i)$ for all time. When we apply locality constraints on the closed-loop maps (i.e. replace $R_c$ and $M_c$ in (\ref{eq:locality_constr}) with $\Phi_x$ and $\Phi_u$), we limit how far a disturbance at a node spreads before it is contained. While both are useful, controller locality tends to be a hard constraint that arises from physical limitations in the communication network, while closed-loop locality is a soft constraint that can be relaxed.

\subsection{Delay}
Let $d(i,j)$ denote the delay from node $j$ to node $i$. In general, $d(i,j)$ is proportional to the distance between nodes $i$ and $j$. Delay constraints are like time-varying locality constraints with an expanding locality, where $\mathcal{L}(i)$ at time $k$ contains all nodes $j$ for which $k \geq d(i,j)$. Delay constraints are enforced as follows:
\begin{equation} \label{eq:delay_constr}
\begin{aligned}
    R_c(k)^{i,j} = 0 \quad \forall k < d(i,j) \\
    BM_c(k)^{i,j} = 0 \quad \forall k < d(i,j) \\
\end{aligned}
\end{equation}
where $B$ is the actuation matrix of the system.

For a system in a chain configuration and $d(i,j)$ proportional to inter-nodal distance, these constraints result in banded diagonal $R_c(k)$ and $M_c(k)$, with wider bands for higher values of $k$.

When we apply delay constraints on the implementation matrices as per (\ref{eq:delay_constr}), we are ensuring that controllers do not require information that cannot be communicated to them in time. For example, node $i$ cannot use any information about node $j$ that is more recent than $t - d(i,j)$. When we apply delay constraints on the closed-loop maps (i.e. replace $R_c$ and $M_c$ in (\ref{eq:delay_constr}) with $\Phi_x$ and $\Phi_u$), we limit how fast a disturbance at node $j$ propagates to the state and input at node $i$. As with locality, the controller delay constraint tends to be a hard constraint arising from physical communication limitations. Unlike in the locality case, the closed-loop delay constraint serves no clear purpose; by separating the controller design from the closed-loop design, we avoid imposing this unnecessary constraint on the closed-loop map.

\subsection{Delay and locality as optimization objectives}

We can augment the objective in (\ref{eq:actual_relax}) with the following terms to encourage tolerance for communication delay:
\begin{equation} \label{eq:delay_tolerance}
    \sum_{k=1}^{T_c} \sum_{i=1}^{n} \sum_{j=1}^{n} e^{dist(i,j)-k} (\|R_c(k)^{i,j}\| + \|BM_c(k)^{i,j}\|)
\end{equation}
where $dist(i,j)$ is the distance between nodes $i$ and $j$ in the network.

We can encourage tolerance for communication locality by using similar terms (note the removal of $k$ from the exponential weight):
\begin{equation}
    \sum_{k=1}^{T_c} \sum_{i=1}^{n} \sum_{j=1}^{n} e^{dist(i,j)} (\|R_c(k)^{i,j}\| + \|BM_c(k)^{i,j}\|)
\end{equation}

Again taking the chain configuration as an example, these terms encourage banded-diagonal $R_c(k)$ and $M_c(k)$ with higher penalties on elements farther away from the diagonal. Elements that survive despite heavy penalty represent edges in the network that require fast communication in order to best preserve the desired closed-loop map.
	\section{EXAMPLES} \label{sec:examples}
All subsequent analysis was done on \text{MATLAB} using the \text{cvx} toolbox with \text{SDPT3} on the low precision setting. The optimization was done on a laptop with an Intel i7 processor and 8GB of RAM.

The system we work with is a 10-node chain with the following tridiagonal $A$ matrix:
\begin{equation}
    A = \begin{bmatrix}
    0.6 & 0.4 & 0 & \ldots & \\
    0.4 & 0.2 & \ddots & & \\
    0 & \ddots & \ddots & \ddots & \\
    \vdots & & \ddots & 0.2 & 0.4 \\
     & & & 0.4 & 0.6 \\
    \end{bmatrix}
\end{equation}

The system has three actuators, located at nodes 3, 6, and 10. The system is marginally stable, with a spectral radius of 1. General observations below extend to larger chains with similarly sparse actuation.

\subsection{Low-norm centralized controllers}
We first synthesize a desired closed-loop map via SLS, with no communication or locality constraints. We use an FIR horizon of $T=20$ and an LQR objective. We then synthesize unconstrained controllers using (\ref{eq:actual_relax}) with an additional $\LOne$ regularization term on ($\Rc$, $\Mc$). We synthesize controllers with order ranging from $T_c=2$ to $T_c=25$.

\begin{figure}[h]
\centering
    \includegraphics[width=8.8cm]{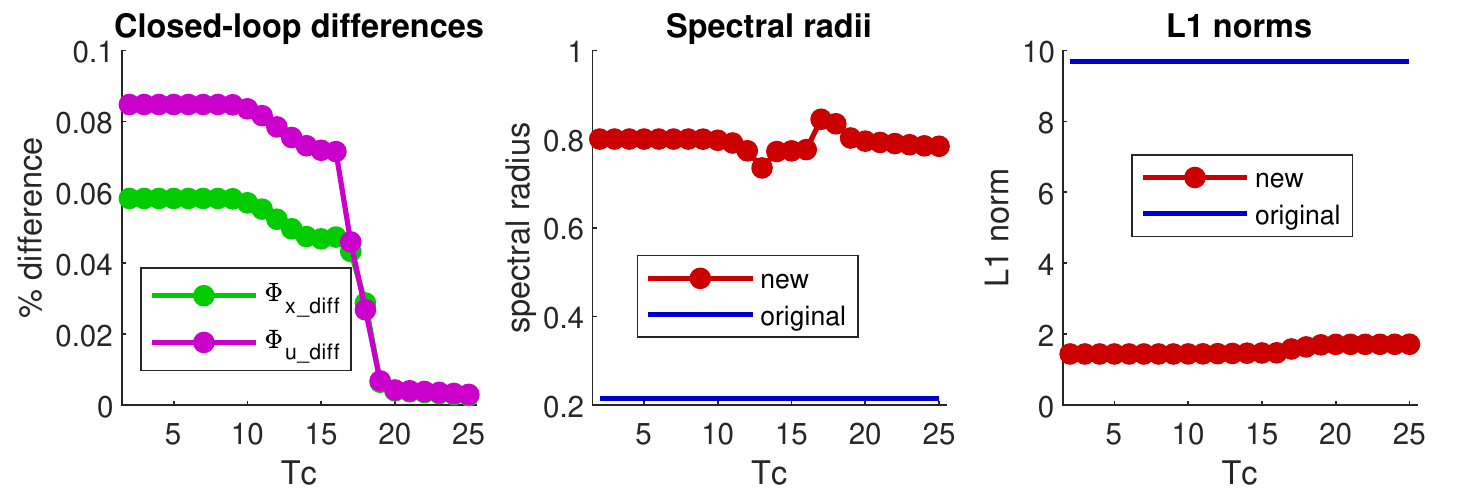}
\caption{Closed-loop differences, spectral radii of internal dynamics, and $\mathcal{L}_1$ norms for controllers with varying $T_c$} \label{fig:example1}
\end{figure}

Fig. \ref{fig:example1} shows the differences between the desired closed loop maps ($\R$, $\M$) and the implemented closed-loop maps ($\Rtilde$, $\Mtilde$), normalized by $\|\R\|$ and $\|\M\|$, respectively. As expected, the closed loop differences decrease with increasing $T_c$. Interestingly, we are able to approximate the system relatively well even for $T_c \ll T$; at $T_c=2$, we are less than 10\% away from the optimal closed-loop map.

Fig. \ref{fig:example1} also shows the spectral radii of $A_z$. The spectral radius of the original controller is far lower than that of the new controllers, suggesting a possible tradeoff between controller norm and internal stability margins. All implementations are internally stable, and spectral radius remains relatively constant over $T_c$.

Lastly, Fig. \ref{fig:example1} shows the $\LOne$ norms of the implementation matrices. All new controllers have significantly lower norm than the original controller, and $\LOne$ norm remains almost constant over $T_c$.

\subsection{Localized LQR controller}
In this example, separating closed-loop synthesis from controller synthesis yields much better results than the original synthesis procedure, in which controller and closed-loop synthesis are coupled.

The objective of this example is to synthesize a controller with an LQR objective and FIR horizon of $T=20$. An SLS formulation of LQR can be found in \cite{Wang2014}. The following constraints must be obeyed: the controller at each node is only allowed to use information from its two neighbouring nodes, and communication speed is restricted to be the same speed as propagation speed.

Directly applying the constraints to the closed-loop map renders the standard SLS problem infeasible (``Constrained CL map'' in Table \ref{table:lqr_tables}); the algorithm cannot find a controller that meets the constraints. We use the virtual localization technique introduced in \cite{Matni2018} to synthesize a controller that meets these constraints (``Virtually local'' in Table \ref{table:lqr_tables}), while relaxing the constraints on the closed-loop map.

We then apply our proposed two-step procedure. First, we synthesize the desired closed-loop maps ($\R$, $\M$) via SLS without communication and locality constraints. We use these closed-loop maps to implement a centralized controller for comparison purposes (``FIR centralized'' in Table \ref{table:lqr_tables}). We then synthesize a controller subject to the communication and locality constraints (``Two-step'' in Table \ref{table:lqr_tables}), using (\ref{eq:actual_relax}) with $\LOne$ regularization. We synthesize one low-order controller with order $T_c=2$, and one full-order controller with $T_c=T$.

For all controllers, we evaluate the LQR cost, spectral radius of the internal dynamics, and $\LOne$ norm of the implementation matrices. The LQR cost is normalized by the optimal infinite horizon LQR cost. Results are shown in Table \ref{table:lqr_tables}.

\begin{table}[htbp]
\caption{Comparison of LQR costs}
\label{table:lqr_tables}
\begin{center}
\begin{tabular}{|l|l|l|l|}
\hline
Controller & LQR cost & Spectral radius & $\LOne$ norm \\
\hline
FIR centralized & 1.001 & 0.214 & 9.688 \\
\hline
Constrained CL map & \multicolumn{3}{c}{\textit{Infeasible}} \vline \\
\hline
Virtually local & 1.294 & 0.847 & 9.704 \\
\hline
Two-step, $T_c=T$ & 1.033 & 0.876 & 1.495 \\
\hline
Two-step, $T_c=2$ & 1.034 & 0.851 & 1.426 \\
\hline
\end{tabular}
\end{center}
\end{table}

In this example, both the full-order and low-order controller (``Two-step'') give an LQR cost increase of about 3$\%$ over the optimal infinite-horizon controller. In contrast, the virtually local controller incurs a cost increase of nearly 30$\%$.

All synthesized controllers are internally stable, with spectral radius less than one. The centralized controller has lower spectral radius than the constrained controllers, which have comparable spectral radii. Additionally, both of our controllers are able to attain an $\LOne$ norm that is very close to the $\LOne$ norm achieved in the previous example, despite much more severe constraints. Overall, our proposed two-step synthesis procedure generates a controller that performs better than the controller generated by existing techniques, without sacrificing internal stability margins.

Interestingly, the low-order controller performs almost as well as the full-order controller, with only 0.1$\%$ performance degradation. This suggests that in this case, highly delayed information (which correspond to higher order terms of the implementation matrices) are not very useful to the controller. 
	\section{CONCLUSIONS AND FUTURE WORK} \label{sec:conclusions}

By separating controller synthesis from closed-loop synthesis, we are able to apply constraints to the controller without unnecessarily limiting the closed-loop map. As demonstrated above, our proposed two-step procedure offers benefits over the original single step procedure. This procedure offers a new perspective on system-level controller design, and an alternative approach for regimes in which standard SLS is infeasible. In future work, we would like to better understand how our method relates to the existing work on virtually localized SLS, and which types of problems each method is better suited to. Additionally, we would like to extend this work to the output feedback case.

Synthesis methods mentioned in this paper can be found in the SLS-MATLAB toolbox at \url{https://github.com/sls-caltech/sls-code}.
	
	\bibliography{library}
	\bibliographystyle{IEEEtran}
\end{document}